\documentclass[11pt,twoside]{article}

\usepackage[ansinew]{inputenc}
\usepackage{amsmath,amsfonts,amssymb,amsthm}
\usepackage{pstricks,pst-node,pst-coil,pst-plot,pstricks-add}
\usepackage{graphics,geometry,epsfig}
\usepackage{bbm}
\usepackage{floatflt}

\newcommand{\field}[1]{\mathbb{#1}}
\newcommand{\N}{\field{N}}
\newcommand{\R}{\field{R}}
\newcommand{\C}{\field{C}}
\newcommand{\Z}{\field{Z}}

\newcommand{\OO}{\mathcal O}

\newcommand{\EE}{\mathcal E}

\newcommand{\eps}{\varepsilon}
\newcommand{\ph}{\varphi}

\newcommand{\norm}[1]{\mbox{$\left\| #1 \right\|$}}
\newcommand{\sprod}[2]{\mbox{$\left\langle #1,#2 \right\rangle$}}
\newcommand{\form}[3]{\mbox{$\left\langle #1\left|#2 \right|#3 \right\rangle $}}

\newcommand{\Rea}{\operatorname{Re}}

\newtheorem{theorem}{Theorem}[section]
\newtheorem{lemma}[theorem]{Lemma}
\newtheorem{corollary}[theorem]{Corollary}

\newtheorem{proposition}[theorem]{Proposition}
\theoremstyle{plain}

\title{Existence of Minimizers \\ in Restricted Hartree--Fock Theory}

\author{
Fabian Hantsch\\
Universit\"at Stuttgart, Fachbereich Mathematik\\
70550 Stuttgart, Germany}

\date{}
\begin{document}
\maketitle
\begin{abstract}
In this note we establish the existence of ground states for atoms within several restricted Hartree--Fock theories. It is shown, for example, that there exists a ground state for closed shell atoms with $N$ electrons and nuclear charge $Z \geq N-1$. This has to be compared with the general Hartree--Fock theory where the existence of a minimizer is known for $Z >N-1$ only.
\end{abstract}

\section{Introduction}

Computations of the electronic structure of atoms and molecules in quantum chemistry in general rely on numerical solutions of simplified versions of the quantum many-body problem at hand. Among those, the Hartree--Fock approximation often serves as a starting point for more accurate approximations such as multi-configuration methods, see for example \cite{Helgaker,SzaboOstlund}. In the simplest version of Hartree--Fock theory the energy is minimized with respect to antisymmetric tensor products of orthonormal one-electron orbitals, the so-called single Slater determinants, and further restrictions are imposed in numerical procedures implementing this variational problem \cite{Cances2003}. In any case the question arises whether a minimizer exists. This paper is concerned with several restricted Hartree--Fock theories for atoms where the one-electron orbitals are products of space and spin wave functions. For each of the considered restrictions we investigate the existence of a minimizer both for neutral atoms and positive ions, as well as for simply charged negative ions.

The existence of a minimizer in the \emph{general Hartree--Fock (GHF)} theory for neutral atoms or positive ions was first established in 1977 by Lieb and Simon \cite{LS1977}. No constraints were imposed in their work besides the orthonormality of the one-electron orbitals. In the meantime there has been remarkable further progress in the study of the variational problem for the Hartree--Fock energy functional. It is known, for example, that there exists a sequence of critical points for this functional \cite{Lions1987}, and convergence properties of various algorithms used for the approximation of critical points were investigated in \cite{CB2000,Cances2000,Levitt}. 

The main concern of this article is the minimization of the Hartree--Fock energy functional under additional constraints.
Our general assumption is that the one-electron states are products of space and spin functions.
First, we treat the \emph{restricted Hartree--Fock (RHF)} functional for closed shell atoms with prescribed angular momentum quantum numbers. Second, we drop the latter requirement, i.e.~we consider atoms with an even number of electrons, where only pairs of spin up and spin down electrons with the same spatial function occur. The corresponding energy functional will be called \emph{spin-restricted Hartree--Fock (SRHF)} functional. It turns out that there exists a ground state in both cases, if $Z \geq N-1$, where $Z$ denotes the nuclear charge and $N$ the number of electrons. The existence of a ground state in the case $Z=N-1$ reminds of the well-known stability of closed shell configurations in chemistry. Third, we look at another restricted Hartree--Fock functional, which is called \emph{unrestricted Hartree--Fock (UHF)} functional in the chemical literature, and must not be confused with the GHF functional. In the UHF setting, the spatial functions corresponding to spin up resp.~spin down functions are chosen independently from each other, but still are assumed to have prescribed angular momenta. In this case a ground state exists if $Z>N-1$, and we provide sufficient conditions under which this is also true if $Z=N-1$. For example, there exists a ground state for $Z=N-1$ in the spinless case (i.e.~if all spins point in the same direction) with two angular momentum shells $\ell_1=0$, $\ell_2 >0$. 

For certain closed shell atoms (e.g.~He, Ne) it is known that the minimization problems for  the general and restricted Hartree--Fock functionals coincide, if $Z \gg N$ \cite{GH2011}. On the other hand there are also cases where they differ \cite{RuskStill}, see \cite{GH2011} for an explanation of this fact.
Nevertheless, the restricted ground states are always critical points of the GHF functional. This is due to the fact that the considered constraints do not require additional Lagrange multipliers in the Euler--Lagrange equations. Thus, this paper also establishes the existence of critical points for the GHF functional in the case $Z=N-1$. 
To our knowledge, the only previous result providing the existence of critical points for the GHF functional in the case $Z=N-1$ is given in the paper \cite{CB2000} of Canc\`es and Le Bris, which in fact even holds for arbitrary $Z>0$. But in general, the critical points constructed in their paper only correspond either to local (not global) minima or saddle points.

In the literature the existence of minimizers for restricted Hartree--Fock functionals has previously been studied for special cases.
Based on Reeken's paper \cite{Reeken1970} on the solutions of the Hartree equation, Bazley and Seydel \cite{BazleySeydel} proved the existence of a minimizer for the spin-restricted Hartree--Fock functional of Helium $(N=2)$, which is given by the restricted Hartree functional. For this functional it is known that there exists a minimizer even if $Z=1=N-1$, see \cite[Theorem II.2]{Lions1987}.
In our paper we extend this result to arbitrary numbers of filled shells.
Lieb and Simon generalize their GHF existence result \cite{LS1977} to certain restricted situations in \cite{LS1974}, but their theorem does not cover the restrictions discussed in this paper. Though, this article has been strongly inspired by their work \cite{LS1977}.
In \cite{Lions1987}, Lions treats restricted Hartree--Fock equations, which arise as the Euler--Lagrange equations of the RHF functional. He proves the existence of a sequence of solutions to these equations provided $Z \geq N$. Lions' proof relies, however, on the unproven assertion that all eigenvalues of a radial Fock operator are simple.
His approach is motivated by the paper of Wolkowisky \cite{Wolkowisky} who shows the existence of solutions for a system of restricted Hartree-type equations. 
A numerical approach to restricted Hartree--Fock theory may be found in the book of Froese Fischer \cite{Froese-Fischer}.
Finally, we mention the article of Solovej \cite{Solovej2003}, where he proves the existence of a universal constant $Q>0$ so that there is no GHF minimizer for $Z \leq N-Q$. 
This establishes the \emph{ionization conjecture} within the Hartree--Fock theory.
The question whether or not there is a GHF minimizer for $Z=N-1$ is open.

The paper is organized as follows: In Section~\ref{sec:closed shell} we introduce the restricted Hartree--Fock functional for closed shell atoms with prescribed angular momentum quantum numbers and prove an existence theorem for minimizers of this functional. The Section~\ref{sec:applications} is devoted to generalizations of the RHF existence theorem to the SRHF and UHF functionals. A derivation of the RHF functional in the closed shell case can be found in Section~\ref{app}. Finally, there is an appendix containing technical lemmas.

\emph{Acknowledgment.} The author thanks M.~Griesemer for drawing his attention to the problem and for helpful discussions. The author is supported by the \emph{Studienstiftung des Deutschen Volkes}.

\section{Minimizers for Closed Shell Atoms}\label{sec:closed shell}

The simplest Hartree--Fock approximation for atoms consists in restricting the admissible $N$-electron states to the set of single Slater determinants, which are of the form
\begin{equation}
\label{SD}
(\ph_1 \wedge \dots \wedge \ph_N)(x_1,\dots,x_N) = \frac{1}{\sqrt{N!}} \sum_{\sigma \in S_N} \operatorname{sgn}(\sigma) \ph_{\sigma(1)}(x_1) \dots \ph_{\sigma(N)}(x_N),
\end{equation}
where $S_N$ denotes the symmetric group of degree $N$, $\operatorname{sgn}(\sigma)$ is the sign of a permutation $\sigma$, and $\ph_1,\dots,\ph_N$ denote orthonormal $L^2(\R^3;\C^2)$-functions. It is well-known, that the energy of an atom with nuclear charge $Z$ and $N$ electrons in the state \eqref{SD} is given by the \emph{general Hartree--Fock (GHF)} functional
\begin{eqnarray}
\label{HF functional}
\EE^{HF}(\ph_1,\dots,\ph_N) = \sum_{j=1}^N \int |\nabla\ph_j|^2- \frac{Z}{|x|}|\ph_j|^2 \,dx 
+ \frac{1}{2} \int\!\!\!\int \frac{\rho(x)\rho(y)-|\tau(x,y)|^2}{|x-y|} \,dx\,dy
\end{eqnarray}
where 
$$
\tau(x,y) := \sum_{j=1}^N \ph_j(x)\overline{\ph_j(y)},  \qquad\qquad \rho(x) := \sum_{j=1}^N |\ph_j(x)|^2
$$
denote the density matrix and the electronic density, respectively.
We write $x=({\bf x}, \mu)\in \R^3 \times \{\pm 1\}$, $\int \,dx$ refers to integration with respect to the product of Lebesgue and counting measure, and $|x-y|=|{\bf x}-{\bf y}|$. 

Given a closed shell atom with $s_0 \in \N$ shells of prescribed angular momentum quantum numbers $\ell_1,\dots,\ell_{s_0} \in \N_0$, we impose the following form on the one-electron orbitals
\begin{equation}
\label{restr form}
\ph_{jm\sigma}({\bf x},\mu) = \frac{f_j(|{\bf x}|)}{|{\bf x}|} Y_{\ell_j m}({\bf x}) \delta_{\sigma \mu}, \quad j=1,\dots,s_0, \ m = -\ell_j,\dots,+\ell_j, \ \sigma = \pm 1,
\end{equation}
where the radial functions $f_j$ are in $L^2(\R_+)$ and 
\begin{equation}\label{constraint}
\langle f_i, f_j \rangle := \int_{\R_+}\overline{f_i}f_j \,dr = \delta_{ij}, \quad \text{if } \ell_i=\ell_j,
\end{equation}
in order to ensure the orthonormality of the functions \eqref{restr form}.
Here $Y_{\ell m}$ denote the usual spherical harmonics.
The Hartree--Fock energy of the Slater determinant built by the orbitals \eqref{restr form} is given by the \emph{restricted Hartree--Fock (RHF)} functional (compare Section~\ref{app}):
\begin{eqnarray}
\lefteqn{\EE^{RHF}(f_1,\dots,f_{s_0}) = 2 \sum_{j=1}^{s_0} (2\ell_j + 1) \left( \int_{\R_+} |f_j'|^2 + \frac{\ell_j(\ell_j+1)}{r^2} |f_j|^2 - \frac{Z}{r} |f_j|^2 \,dr \right)} \nonumber \\
& & + \frac{1}{2} \sum_{j,k=1}^{s_0} (2\ell_j+1)(2\ell_k+1) \Bigg( \int\!\!\!\int_{(\R_+)^2} 4 \frac{|f_j(r)|^2 |f_k(s)|^2}{\max\{r,s\}} \nonumber \\
& & \qquad  \qquad\qquad\qquad\qquad\qquad- 2 \overline{f_j(r)f_k(s)}U_{\ell_j \ell_k}(r,s) f_k(r) f_j(s) \,dr\,ds  \Bigg). \label{rest energy}
\end{eqnarray}
The integral kernels $U_{\ell_j\ell_k}$ appearing in the last term on the right-hand side are given in \eqref{Uform}. We shall only need their properties collected in Lemma~\ref{lm:U-ell}.

Let $H^1_0(\R_+)$ denote the completion of $C^\infty_0(\R_+)$ with respect to the $H^1(\R_+)$-norm. 
The RHF functional \eqref{rest energy} is bounded below, if the functions $f_1,\dots,f_{s_0}$ are in $H^1_0(\R_+)$ and obey the constraints \eqref{constraint}, see Lemma~\ref{lm:technical}.
We define the RHF ground state energy by
\begin{equation}
\label{min prob}
E(N,Z) = \inf\{\EE^{RHF}(f_1,\dots,f_{s_0}) | f_1,\dots, f_{s_0} \in H^1_0(\R_+), \sprod{f_i}{f_j} = \delta_{ij} \text{ if } \ell_i= \ell_j\},
\end{equation}
where the dependence of $E(N,Z)$ on $\ell_1,\dots,\ell_{s_0}$ is suppressed. 
The main question of this paper is whether the infimum in \eqref{min prob} is actually a minimum.

If there exist minimizing functions $f_1,\dots,f_{s_0}$ obeying the constraints of \eqref{min prob}, then they are solutions of the corresponding Euler--Lagrange equations, which we may assume to have the form (see Remark (b) below)
\begin{equation}\label{EL eqn}
H_{\ell_i}f_i = \eps_i f_i, \qquad i=1,\dots,s_0,
\end{equation}
with \emph{radial Fock operators} given by
\begin{eqnarray*}
H_{\ell_i} & = & -\partial^2_r + \frac{\ell_i(\ell_i+1)}{r^2} -\frac{Z}{r} +2U - K_{\ell_i}, \qquad i=1,\dots,s_0, \quad \text{ where} \\
(Uf)(r) & = & \sum_{j=1}^{s_0} (2\ell_j+1) \int_{\R_+} \frac{|f_j(s)|^2}{\max\{r,s\}} \,ds f(r), \\
(K_{\ell}f)(r) & = & \sum_{j=1}^{s_0} (2\ell_j+1) f_j(r) \int_{\R_+} \overline{f_j(s)}f(s) U_{\ell \ell_j}(r,s) \,ds.
\end{eqnarray*}
We suppress the dependence of the operators $U$, $K_{\ell}$ and thus $H_{\ell_i}$ on the functions $f_1,\dots,f_{s_0}$. 
The Euler--Lagrange equations \eqref{EL eqn}, called \emph{Hartree--Fock equations}, form a set
of $s_0$ coupled non-linear eigenvalue equations for the functions $f_1,\dots,f_{s_0}$.

\pagebreak

\noindent \emph{Remarks.}
\begin{itemize}
\item[(a)] By Lemma~\ref{lm:technical}, the operators $H_{\ell_i}$ are symmetric semi-bounded operators on $C_0^\infty(\R_+)$. Therefore, minimizing functions $f_1,\dots,f_{s_0}$ obeying the constraints of \eqref{min prob} are in the domain $D(H_{\ell_i})$ of the Friedrichs extension of $H_{\ell_i}$, which is contained in $H^1_0(\R_+)$.
\item[(b)] The Euler--Lagrange equations for minimizing functions $f_1,\dots,f_{s_0}$ obeying the constraints of \eqref{min prob} are given by $H_{\ell_i}f_i = \sum_j \eps_{ij}f_j$, where the sum runs over all indices $j$ with $\ell_j = \ell_i$.
Since the functional $\EE^{RHF}$ is invariant under unitary transformations of the subspaces of $L^2(\R_+)$ spanned by all radial functions $f_j$ with equal angular momentum quantum numbers, the minimizing functions $f_1,\dots,f_{s_0}$ can always be chosen as eigenfunctions of the radial Fock operators.
\item[(c)]
The constraints of \eqref{min prob} may be relaxed without lowering the ground state energy, more precisely
$E(N,Z)=\tilde E(N,Z)$ for
\begin{eqnarray}
\tilde E (N,Z) & = & \inf \{\EE^{RHF}(f_1,\dots,f_{s_0}) | f_1,\dots, f_{s_0} \in H_0^1(\R_+), \sprod{f_i}{f_j} = 0 \text{ if } \ell_i = \ell_j \nonumber \\
& & \qquad \qquad \qquad\qquad \quad\quad \text{ and } i \neq j,
\norm{f_i} \leq 1 \text{ for all } i \}. \label{generalized minimization problem}
\end{eqnarray}
This can be seen using similar arguments as for the general Hartree--Fock functional in \cite[section II.2]{Lions1987}.
The following theorem shows that the relaxed minimization problem always possesses a minimizer.
\end{itemize}

\begin{theorem}
\label{thm:existence}
Let $s_0\in\N$, $\ell_1,\dots,\ell_{s_0} \in \N_0$, and $Z>0$. Then, there exist functions $f_1,\dots,f_{s_0} \in H_0^1(\R_+)$, which minimize the RHF functional \eqref{rest energy} under the constraints
\begin{eqnarray*}
\sprod{f_i}{f_j} & = & 0 \text{ if } \ell_i = \ell_j \text{ and } i \neq j, \\
\norm{f_i} & \leq & 1 \text{ for all } i.
\end{eqnarray*}
Moreover, $f_i \in D(H_{\ell_i})$, $H_{\ell_i}f_i = \eps_i f_i$, and:
\begin{itemize}
\item[(i)] Either $\eps_i \leq 0$ or $f_i=0$. $\eps_i < 0$ implies $\norm{f_i} = 1$.
\item[(ii)] If $Z>N- 2(2\ell_i+1)$, then $f_i \neq 0$.\\ If $Z \geq N-1$, then $\|f_i\|=1$ for all $i=1,\dots,s_0$.
\item[(iii)] If $Z>N-1$, then $\eps_i<0$ and $\norm{f_i}=1$ for all $i=1,\dots,s_0$.
\end{itemize}
\end{theorem}

\noindent \emph{Remarks.} 
\begin{itemize}
\item[(a)] Theorem~\ref{thm:existence}~(ii) shows that for $Z=N-1$ there always exists a normalized minimizer for $\EE^{RHF}$. In this case we do not know whether or not $\eps_i<0$. Nevertheless, it is clear that $Z>N-2(2\ell_i+1)$ always implies $E(N,Z)<E^{(i)}(N-2(2\ell_i+1),Z)$ for all $i=1,\dots,s_0$, where $E^{(i)}(N-2(2\ell_i+1),Z)$ denotes the minimal energy in the case where all electrons of the $i$-th shell are dropped.
\item[(b)] In general Hartree--Fock theory it is known that the minimizing functions can be chosen as eigenfunctions to the $N$ lowest eigenvalues of the corresponding Fock operator. Moreover, there is a gap between the occupied and unoccupied eigenvalues \cite{BLLS1994}. 
It would be interesting to know whether similar results hold also in the restricted Hartree--Fock theory, where, unfortunately, the method of \cite{BLLS1994} is not applicable.
\end{itemize}

Before turning to the proof of Theorem~\ref{thm:existence} we introduce the following notations that will be used throughout the paper:
$$
r_> := \max\{r,s\} , \qquad r_< := \min\{r,s\} , \qquad \text{for } r,s \geq 0.
$$
We write $\EE^{RHF}(f_1,\dots, \hat f_i, \dots,f_{s_0})$ to denote the restricted Hartree--Fock functional where the electrons of the $i$-th shell are dropped. The following lemma exhibits the dependence of $\EE^{RHF}(f_1,\dots,f_i,\dots,f_{s_0})$ on $f_i$, and will be crucial for the existence of a minimizer in the critical case $Z=N-1$. It follows easily from the definition of $\EE^{RHF}$ if we set $P_i(r,s) := (2\ell_i+1)(2 r_>^{-1} - U_{\ell_i \ell_i}(r,s) )$.

\begin{lemma}
\label{lm:decomp}
Let $s_0 \in \N$, $\ell_1,\dots,\ell_{s_0} \in \N_0$, $Z>0$ and $f_1,\dots, f_{s_0} \in H^1_0(\R_+)$. Furthermore, let $i \in \{1,\dots,s_0\}$ and let $H_{\ell_i}^{(i)}$ denote the Fock operator where all electrons of the $i$-th shell are dropped. Then:
\begin{eqnarray}
\EE^{RHF}(f_1,\dots,f_i, \dots,f_{s_0}) & = & \EE^{RHF}(f_1,\dots,\hat f_i, \dots, f_{s_0}) + 2(2\ell_i+1) \form{f_i}{H_{\ell_i}^{(i)}}{f_i} \nonumber \\
& &  +  (2\ell_i+1) \form{f_i \otimes f_i}{P_i}{f_i \otimes f_i}, \label{decomp1}
\end{eqnarray}
where $P_i(r,s)=P_i(s,r)$ and 
$$
\frac{2\ell_i+1}{\max\{r,s\}} \leq P_i(r,s) \leq \frac{4\ell_i+1}{\max\{r,s\}}, \qquad r,s \geq 0.
$$
Furthermore, for all $\lambda \geq 0$:
\begin{eqnarray}
\lefteqn{\EE^{RHF}(f_1, \dots, \frac{f_i+\delta h}{\sqrt{1+\lambda\delta^2}},\dots,f_{s_0})} \nonumber \\
& = & \EE^{RHF}(f_1,\dots,f_i,\dots,f_{s_0}) + 4(2\ell_i+1)\delta\Rea\form{h}{H_{\ell_i}}{f_i} \nonumber \\
& & + 2(2\ell_i+1)\delta^2 \Big( \form{h}{H_{\ell_i}^{(i)}}{h} - \lambda \form{f_i}{H_{\ell_i}}{f_i} + \Rea\form{h \otimes h}{P_i}{f_i \otimes f_i} \nonumber \\
& & \qquad\qquad\qquad\quad + \form{f_i \otimes h + h \otimes f_i}{P_i}{f_i \otimes h} \Big) + \OO(\delta^3) \label{decomp2}
\end{eqnarray}
for $\delta \to 0$.
\end{lemma}

\begin{proof}[Proof of Theorem~\ref{thm:existence}]
First, we give a proof of the existence of a minimizer for the relaxed minimization problem, which proceeds the same way as in the paper of Lieb and Simon \cite{LS1977}.
$\EE^{RHF}(g_1,\dots,g_{s_0})$ is bounded below independently of $g_1,\dots,g_{s_0} \in H^1_0(\R_+)$ with $\|g_i\|\leq1$, see Lemma~\ref{lm:technical}.
Thus, let $g_1^{(n)},\dots,g_{s_0}^{(n)}$ be a minimizing sequence for the relaxed minimization problem \eqref{generalized minimization problem}.
Again by Lemma~\ref{lm:technical}, $(g_j^{(n)})_{n \in \N}$ is bounded in $H^1_0(\R_+)$.
Hence, there exist weakly-$H_0^1(\R_+)$ convergent subsequences $g_j^{(n)} \rightharpoonup g_j$ $(n \to \infty)$, $j=1,\dots,s_0$.
Fix $i \in \{1,\dots,s_0\}$. Without loss of generality we may assume that $g_1,\dots,g_{k_i}$ are all functions $g_j$ with $\ell_j=\ell_i$.
The matrix $M:=(\langle g_j, g_k \rangle)_{j,k=1,\dots,k_i}$ is hermitian and obeys $0 \leq M\leq 1$ (c.f.~\cite[Lemma~2.2]{LS1977}), so there exists a unitary $k_i \times k_i$ matrix $U$ with the property $U^* M U = D$, where $D$ is a diagonal matrix with eigenvalues in $[0,1]$.
If we define $f_j = \sum_{k=1}^{k_i} u_{kj}g_k$, $j=1,\dots,k_i$, then $\langle f_j, f_k \rangle = \lambda_{j}\delta_{jk}$, $0 \leq \lambda_j \leq 1$.
It is easy to see that $\EE^{RHF}$ is invariant under such transformations.
Thus, transforming each subspace of functions with equal angular momentum quantum numbers in this way, we obtain functions $f_1,\dots,f_{s_0}$ with $\langle f_i, f_j \rangle =0$, if $\ell_i=\ell_j$, $i \neq j$, $\norm{f_i} \leq 1$ for all $i$.
Furthermore, $f_1,\dots,f_{s_0}$ minimize $\EE^{RHF}$, because
\begin{eqnarray*}
\tilde E(N,Z) & \leq & \EE^{RHF}(f_1,\dots,f_{s_0}) = \EE^{RHF}(g_1,\dots,g_{s_0}) \\
& \leq & \liminf_{n \to \infty} \EE^{RHF}(g_1^{(n)},\dots,g_{s_0}^{(n)}) = \tilde E(N,Z),
\end{eqnarray*}
where we used Lemma~\ref{lm:technical}.
By further transformations we can achieve that $f_1,\dots,f_{s_0}$ are eigenfunctions of the operators $H_{\ell_i}$.

(i) Let $f_i\neq 0$ and assume that $\eps_i > 0$. Then, by \eqref{decomp2} with $\lambda=0$ and $h=f_i$, the energy decreases if we decrease the norm of $f_i$. Let $\eps_i <0$ and assume that $\norm{f_i}<1$. Then, the energy is decreased by increasing the norm of $f_i$.

(ii) We prove the following more general statement: 
Let $0 \leq \mu \leq 1$ and let \linebreak $Z \geq N-1-(1-\mu)(4\ell_i+1)$, then $\mu \leq \|f_i\|^2 \leq 1$.

There is nothing to prove in the case $\mu=0$. Therefore, let $\mu >0$ and assume that $\|f_i\|^2 < \mu$.
We show that there exists $h \in H^1_0(\R_+)$ with $h \perp f_j$, if $\ell_j=\ell_i$, such that 
$$
\EE^{RHF}(f_1,\dots, f_i+\delta h, \dots, f_{s_0}) < \EE^{RHF}(f_1,\dots, f_i, \dots, f_{s_0})
$$
for small $\delta\neq 0$, which contradicts the minimization property of $f_1,\dots,f_{s_0}$.
The dependence of the left-hand side on $h \in H^1_0(\R_+)$ is given by \eqref{decomp2} with $\lambda=0$. The factor of $\delta$ in \eqref{decomp2} vanishes since $f_1,\dots,f_{s_0}$ is a minimizer. Therefore, we only have to show that there exist infinitely many normalized functions $h \in H^1_0(\R_+)$ with disjoint supports, such that 
the factor of $\delta^2$ in \eqref{decomp2}
\begin{equation}
\label{delta2term}
\form{h}{H_{\ell_i}^{(i)}}{h} + \form{f_i \otimes h}{P_i}{f_i \otimes h} + \form{f_i \otimes h}{P_i}{h \otimes f_i} + \operatorname{Re}\form{h \otimes h}{P_i}{f_i \otimes f_i}
\end{equation}
is negative.
We may drop the $\operatorname{Re}$-term because it becomes non-positive upon a suitable choice of the phase of $h$. 
Let $J \in C_0^\infty(\R_+)$, $\operatorname{supp}(J) \subset [1,2]$, $\norm{J} =1$. Furthermore, we define $J_R(r):= R^{-1/2} J(r/R)$ for $R>0$, then $\operatorname{supp}(J_R) \subset [R,2R]$, $\norm{J_R}=1$, $J_R \in C_0^\infty(\R_+)$.
Using $U(r) \leq r^{-1} \sum_{j=1}^{s_0} (2\ell_j+1)$ and $K_\ell \geq 0$ (Lemma~\ref{lm:technical}), we see that 
\begin{eqnarray}
\label{eqn inf}
\form{J_R}{H_{\ell_i}^{(i)}}{J_R} & \leq & \form{J_R}{\displaystyle -\partial_r^2 + \frac{\ell_i(\ell_i+1)}{r^2}- \frac{Z}{r} + \frac{N-2(2\ell_i+1)}{r}}{J_R}.
\end{eqnarray}
This inequality combined with the estimate for $P_i$ in Lemma~\ref{lm:decomp} allows us to estimate \eqref{delta2term} with the choice $h=J_R$
\begin{eqnarray*}
\form{J_R}{H_{\ell_i}^{(i)}}{J_R} & \leq & \frac{1}{R^2} \form{J}{-\partial_r^2 + \frac{\ell_i(\ell_i+1)}{r^2}}{J} - \frac{(4\ell_i +1)\mu}{R} \form{J}{\frac{1}{r}}{J}, \\
\form{f_i \otimes J_R}{P_i}{f_i \otimes J_R} & \leq & \frac{(4\ell_i +1)\norm{f_i}^2}{R} \form{J}{\frac{1}{r}}{J}, \\
\form{f_i \otimes J_R}{P_i}{J_R \otimes f_i} & = & o\left(\frac{1}{R}\right)
\end{eqnarray*}
for $R \to \infty$.
The sum of the three terms on the right-hand side becomes negative for $R$ large enough, because $\norm{f_i}^2 < \mu$, by assumption. This proves (ii).

(iii) It suffices to show that $\eps_j<0$, $j=1,\dots,s_0$, see (i) and (ii). Assume that $\eps_i=0$. 
We show that there exists $h\in H^1_0(\R_+)$, $\norm{h}=1$, $h \perp f_j$, if $\ell_i = \ell_j$, so that
$$
\EE^{RHF}(f_1,\dots,\frac{f_i + \delta h}{\sqrt{1 + \delta^2}}, \dots, f_{s_0}) < \EE^{RHF}(f_1,\dots, f_i, \dots, f_{s_0})
$$
for small $\delta \neq 0$. Again, the dependence on $h$ of the left-hand side is given by \eqref{decomp2} with $\lambda=1$. Since $\eps_i=0$, it suffices to show that the factor of $\delta^2$, which is the same as in \eqref{delta2term}, can be made negative by suitable choices of $h$. This can be done choosing the same scaled functions as in (ii), but now using $Z > N-1$ instead of $\norm{f_i}^2 < \mu$.
\end{proof}

\section{Other Restricted Hartree--Fock Functionals}\label{sec:applications}

Theorem~\ref{thm:existence} can be readily generalized to other restricted Hartree--Fock functionals. In this section we present analogous results for a so-called UHF functional as well as for a spin-restricted Hartree--Fock functional.
In the former case we continue assuming that the electrons are in product states of space and spin but we drop the condition that the spatial wave functions for spin up resp.~spin down electrons are equal in each shell with fixed angular momentum quantum numbers. More precisely, we consider electrons that are in states of the form
\begin{eqnarray}
\label{restr form2 filled}
\ph_{jm\uparrow}({\bf x},\mu) & = & \frac{f_j^\alpha(|{\bf x}|)}{|{\bf x}|} Y_{\ell_j^\alpha m}({\bf x}) \delta_{\mu,+1}, \quad j=1,\dots,s_0^\alpha, \ m = -\ell_j^\alpha,\dots,+\ell_j^\alpha, \\
\label{restr form2 half-filled}
\ph_{jm\downarrow}({\bf x},\mu) & = & \frac{f_j^\beta(|{\bf x}|)}{|{\bf x}|} Y_{\ell_j^\beta m}({\bf x}) \delta_{\mu,-1}, \quad j=1,\dots,s_0^\beta, \ m = -\ell_j^\beta,\dots,+\ell_j^\beta,
\end{eqnarray}
where $s_0^\alpha, s_0^\beta \in \N_0$, $\ell_1^\alpha,\dots,\ell_{s_0^\alpha}^\alpha,\ell_1^\beta,\dots,\ell_{s_0^\beta}^\beta \in \N_0$, and for all $\nu \in \{\alpha,\beta\}$, $i,j \in \{1,\dots,s_0^\nu\}$
$$
f_i^\nu \in H^1_0(\R_+), \qquad \langle f_i^\nu, f_j^\nu \rangle = \delta_{ij}, \text{ if } \ell_i^\nu = \ell_j^\nu.
$$
The corresponding Hartree--Fock functional, which is called \emph{unrestricted Hartree--Fock (UHF)} functional, takes the form
\begin{eqnarray}
\lefteqn{\EE^{UHF}(f_1^\alpha,\dots,f_{s_0^\alpha}^\alpha;f_1^\beta,\dots,f_{s_0^\beta}^\beta)}\nonumber \\ 
& = & \sum_{\nu \in \{\alpha,\beta\}} \sum_{j=1}^{s_0^\nu} (2\ell_j^\nu + 1) \form{f_j^\nu}{\displaystyle -\partial_r^2 + \frac{\ell_j^\nu(\ell_j^\nu+1)}{r^2} - \frac{Z}{r}}{f_j^\nu} \nonumber \\
& & + \frac{1}{2} \sum_{\nu \in \{\alpha,\beta\}} \sum_{j,k=1}^{s_0^\nu} \left( D[f_j^\nu,f_k^\nu]- E[f_j^\nu,f_k^\nu]\right) + \sum_{j=1}^{s_0^\alpha}\sum_{k=1}^{s_0^\beta} D[f_j^\alpha,f_k^\beta]. \label{uhf energy}
\end{eqnarray}
Here we use the shorthand notations 
\begin{eqnarray*}
D[f_j^\nu,f_k^\mu] & := & (2\ell_j^\nu+1)(2\ell_k^\mu+1)\form{f_j^\nu \otimes f_k^\mu}{\frac{1}{r_>}}{f_j^\nu \otimes f_k^\mu}, \\
E[f_j^\nu,f_k^\mu] & := & (2\ell_j^\nu+1)(2\ell_k^\mu+1)\form{f_j^\nu \otimes f_k^\mu}{U_{\ell_j^\nu \ell_k^\mu}}{f_k^\mu \otimes f_j^\nu}.
\end{eqnarray*}
Given $\nu\in\{\alpha,\beta\}$ and $\ell \in \N_0$ we introduce Fock operators
\begin{eqnarray*}
H_\ell^\nu & := & -\partial_r^2 + \ell(\ell+1)r^{-2} - Zr^{-1} + U - K^\nu_\ell, \quad \text{where} \\
(Uf)(r) & = & \sum_{\kappa\in\{\alpha,\beta\}}\sum_{j=1}^{s_0^\kappa} (2\ell_j^\kappa + 1) \int_{\R_+} \frac{|f_j^\kappa(s)|^2}{\max\{r,s\}}\,ds f(r), \\
(K_\ell^\nu f)(r) & = & \sum_{j=1}^{s_0^\nu} (2\ell_j^\nu + 1)f_j^\nu(r) \int_{\R_+} \overline{f_j^\nu(s)} U_{\ell\ell_j^\nu}(r,s) f(s) \,ds
\end{eqnarray*}
for $f\in L^2(\R_+)$. Again these operators depend on the functions $f_1^\alpha,\dots,f_{s_0^\beta}^\beta$.
Using the same methods as in the proof of Theorem~\ref{thm:existence}, the following existence theorem can be proved:
\begin{theorem}
\label{thm:existence general}
Let $s_0^\alpha, s_0^\beta \in \N_0$, $\ell_1^\alpha,\dots,\ell_{s_0^\alpha}^\alpha,\ell_1^\beta,\dots,\ell_{s_0^\beta}^\beta \in \N_0$, and $Z>0$. Then, there exist functions 
$f_1^\alpha,\dots,f_{s_0^\alpha}^\alpha,f_1^\beta,\dots,f_{s_0^\beta}^\beta \in H^1_0(\R_+)$, which minimize the UHF functional \eqref{uhf energy} under the constraints: for all $\nu\in\{\alpha,\beta\}$ and $i,j \in \{1,\dots,s_0^\nu\}$
\begin{eqnarray*}
\langle f_i^\nu, f_j^\nu \rangle & = & 0 \text{ if } \ell_i^\nu = \ell_j^\nu, \ i \neq j, \\
\norm{f_i^\nu} & \leq & 1 \text{ for all } i.
\end{eqnarray*}
Moreover, $f_i^\nu \in D(H_{\ell_i^\nu}^\nu)$, $H_{\ell_{i}^\nu}^\nu f_i^\nu = \eps^\nu_i f_i^\nu$.
\begin{itemize}
\item[(i)] Either $\eps_i^\nu \leq 0$ or $f_i^\nu =0$. $\eps_i^\nu <0$ implies $\|f_i^\nu\| =1$.
\item[(ii)] 
If $Z > N-(2\ell_i^\nu+1)$, then $f_i^\nu \neq 0$. \\ If $Z\geq N-1$ and $\ell_i^\nu \neq 0$, then $\| f_i^\nu \| =1$.
\item[(iii)] If $Z>N-1$, then $\eps_i^\nu <0$ and $\norm{f_i^\nu}=1$ for all $\nu \in \{\alpha,\beta\}, \ i=1,\dots,s_0^\nu$.
\end{itemize}
\end{theorem}

\noindent\emph{Remarks}.
\begin{itemize}
\item[(a)] We do not know, except for the case where $\ell=0$, whether the occupied eigenvalues of the corresponding Fock operator are the lowest eigenvalues or whether there is a gap between occupied and unoccupied eigenvalues.
\item[(b)] In general, Theorem~\ref{thm:existence general} does not imply the existence of UHF minimizers in the case of $Z=N-1$. Nevertheless, in the special case where all spins point in the same direction (i.e.~the spinless case) the following existence result holds true.
\end{itemize}

\begin{corollary}
\label{existence of negative ions}
Let $s_0^\alpha \in \N$, $s_0^\beta=0$, and let $\ell_1^\alpha=0$, $\ell_2^\alpha,\dots,\ell_{s_0^\alpha}^\alpha>0$ with
$$
s_0^\alpha < 2 + \sum_{i=2}^{s_0^\alpha} \left( \frac{\ell_i^\alpha}{\ell_i^\alpha+1} \right)^2.
$$ 
If $Z=\sum_{i=2}^{s_0^\alpha} (2\ell_i^\alpha+1)$ and $N= Z+1$, then the UHF functional \eqref{uhf energy} has a minimizer under the constraints $\langle f_i^\alpha, f^\alpha_j \rangle = \delta_{ij}$ for all $i,j = 1,\dots, s_0^\alpha$ with $\ell_i=\ell_j$.
\end{corollary}

\noindent\emph{Remark.} The condition of Corollary~\ref{existence of negative ions} always holds in the case of two shells $s_0^\alpha=2$, $\ell_1^\alpha=0$, $\ell_2^\alpha>0$.

\begin{proof}[Proof of Corollary~\ref{existence of negative ions}]
Theorem~\ref{thm:existence general} yields the existence of $f_1^\alpha,\dots, f_{s_0^\alpha}^\alpha \in H^1_0(\R_+)$, which minimize \eqref{uhf energy} under the constraints $\langle f_i^\alpha, f_j^\alpha \rangle = 0$ if $\ell_i^\alpha = \ell_j^\alpha$ and $i \neq j$, $\| f_i^\alpha \| \leq 1$ for all $i$. Clearly, $\| f_2^\alpha \|=\dots=\| f_{s_0^\alpha}^\alpha\|=1$ by (ii). Observe that
\begin{equation}\label{cor:eqn1}
\EE^{UHF}(f_1^\alpha,\dots,f_{s_0^\alpha}^\alpha) \leq \inf_{\begin{array}{c} \scriptstyle g \in H^1_0(\R_+), \\ \scriptstyle \|g\| \leq 1 \end{array}} \EE^{UHF}(g,0,\dots,0) = - \frac{Z^2}{4},
\end{equation}
and on the other hand
\begin{eqnarray}
\EE^{UHF}(0,f_2^\alpha,\dots,f_{s_0^\alpha}^\alpha) & \geq & - \frac{Z^2}{4} \sum_{i=2}^{s_0^\alpha} \frac{2\ell_i^\alpha+1}{(\ell_i^\alpha+1)^2} = -\frac{Z^2}{4} \left( s_0^\alpha-1 - \sum_{i=2}^{s_0^\alpha} \left(\frac{\ell_i^\alpha}{\ell_i^\alpha+1}\right)^2 \right) \nonumber \\
& > & -\frac{Z^2}{4}, \label{cor:eqn2}
\end{eqnarray}
where we dropped the electron--electron energy and estimated the remaining terms by the hydrogen ground state energies in the first inequality, and used the condition on $s_0^\alpha$ in the second inequality.
Assume, that $\langle f_1^\alpha |H_0^\alpha | f_1^\alpha\rangle = 0$, then $$ \EE^{UHF}(f_1^\alpha,\dots,f_{s_0^\alpha}^\alpha)=\EE^{UHF}(0,f_2^\alpha,\dots,f_{s_0^\alpha}^\alpha),$$
because $\EE^{UHF}(f_1^\alpha,\dots,f_{s_0^\alpha}^\alpha) = \EE^{UHF}(0,f_2^\alpha,\dots,f_{s_0^\alpha}^\alpha) + \langle f_1^\alpha|H_0^\alpha |f_1^\alpha \rangle$, which contradicts \eqref{cor:eqn1} and \eqref{cor:eqn2}.
Therefore, $\langle f_1^\alpha | H_0^\alpha | f_1^\alpha \rangle = \eps_1^\alpha \| f_1^\alpha \|^2 <0$, which implies $\eps_1^\alpha < 0$ and thus $\norm{f_1^\alpha}=1$.
\end{proof}

Another frequently used model for atoms with an even number of electrons is the \emph{spin-restricted Hartree--Fock (SRHF)} model \cite{Cances2003}. It emerges from the RHF model in Section~\ref{sec:closed shell} by dropping the prescribed angular momentum quantum numbers.
More precisely, for an atom with atomic number $Z$ and $N=2n$ we impose the following form on the one-electron orbitals
$$
\ph_{i\sigma}({\bf x},\mu) = \ph_i({\bf x})\delta_{\sigma\mu}, \ i=1,\dots,n, \ \sigma = \pm 1,
$$
where $\ph_i \in H^1(\R^3)$ and $\langle \ph_i, \ph_j \rangle := \int_{\R^3} \overline{\ph_i}\ph_j \,d{\bf x} = \delta_{ij}$. Then the restricted Hartree--Fock functional reads
\begin{eqnarray}
\lefteqn{\EE^{SRHF}(\ph_1,\dots,\ph_n) = 2 \sum_{i=1}^n \int |\nabla \ph_i({\bf x})|^2 - \frac{Z}{|{\bf x}|}|\ph_i({\bf x})|^2  \,d{\bf x}} \nonumber \\
&  & \qquad\qquad\qquad\quad  + \frac{1}{2} \int \!\!\!\int 4 \frac{\rho({\bf x})\rho({\bf y})}{|{\bf x}-{\bf y}|} - 2 \frac{|\tau({\bf x},{\bf y})|^2}{|{\bf x}-{\bf y}|} \,d{\bf x}\,d{\bf y}. \label{srhf energy}
\end{eqnarray}
Here the electronic density matrix and the electronic density are given by
$$
\tau({\bf x},{\bf y}) = \sum_{i=1}^n \ph_i({\bf x})\overline{\ph_i({\bf y})}, \quad \rho({\bf x}) = \sum_{i=1}^n |\ph_i({\bf x})|^2.
$$
The corresponding Fock operator is given by
$$
H = -\Delta - \frac{Z}{|{\bf x}|} + 2 \int \frac{\rho({\bf y})}{|{\bf x}-{\bf y}|} \,dy - K,
$$
where $(K\ph)({\bf x}) :=  \int \frac{\tau({\bf x},{\bf y})\ph({\bf y})}{|{\bf x}-{\bf y}|} \,d{\bf y}$. 
Using similar ideas as in the proof of Theorem~\ref{thm:existence} the following existence theorem holds true for the spin-restricted Hartree--Fock functional:

\begin{theorem}
\label{thm:existence spin}
Let $Z>0$ and $N=2n$. Then, there exist functions $\ph_1,\dots,\ph_n \in H^1(\R^3)$, which minimize the SRHF functional \eqref{srhf energy} under the constraints
\begin{eqnarray*}
\sprod{\ph_i}{\ph_j} & = & 0 \text{ if } i \neq j, \\
\norm{\ph_i} & \leq & 1 \text{ for all } i.
\end{eqnarray*}
Moreover, $\ph_i \in D(H)=H^2(\R^3)$, $H \ph_i = \eps_i \ph_i$, and:
\begin{itemize}
\item[(i)] Either $\eps_i \leq 0$ or $\ph_i=0$. $\eps_i < 0$ implies $\norm{\ph_i} = 1$.
\item[(ii)] 
If $Z>N-2$, then $\ph_i \neq 0$ for all $i=1,\dots,n$. \\ If $Z\geq N-1$, then $\| \ph_i \|=1$ for all $i=1\dots,n$.
\item[(iii)] If $Z>N-1$, then $\eps_i<0$ and $\norm{\ph_i}=1$ for all $i=1,\dots,n$.
\end{itemize}
\end{theorem}

\noindent\emph{Remark.} For this spin-restricted Hartree--Fock functional the minimizer exists for all $Z \geq N-1$. Again we do not know whether or not $\eps_j$ are the $n$ lowest eigenvalues of $H$, although there seem to be no numerical counterexamples \cite{Cances2003}.

\section{Derivation of the Closed Shell Energy Functional}\label{app}

For the reader's convenience we give here a self-contained derivation of the restricted Hartree--Fock functional \eqref{rest energy}. For this purpose, we begin with a lemma that will be useful for the calculation of the electron--electron interaction energy.

Let $P_\ell$ denote the $\ell$-th Legendre polynomial. We remark that for ${\bf \hat x},{\bf \hat y} \in \mathbb{S}^2$ and $\ell \in \N_0$ the addition theorem 
\begin{equation}\label{LegProp}
\sum_{m=-\ell}^\ell Y_{\ell m}({\bf \hat x}) \overline{Y_{\ell m}({\bf \hat y})} = \frac{2\ell +1}{4\pi} P_\ell({\bf \hat x} \cdot {\bf \hat y})
\end{equation}
holds, where ${\bf \hat x} \cdot {\bf \hat y}$ is the usual scalar product of two vectors in $\R^3$.
\begin{proposition}
\label{prop:representation formula}
Let $\ell, L \in \N_0$ and $M \in \Z$, $|M| \leq L$. Then for all $r,s>0$ and ${\bf \hat x} \in \mathbb{S}^2$:
\begin{equation}
\label{eqn general}
\frac{1}{4\pi}\int_{\mathbb{S}^2} \frac{P_\ell({\bf \hat x} \cdot {\bf \hat y})Y_{LM}({\bf \hat y})}{|r{\bf \hat x}-s{\bf \hat y}|} \,d\sigma({\bf \hat y}) = 
Y_{LM}({\bf \hat x}) \sum_{n=|L-\ell|}^{L+\ell} \begin{pmatrix} L & \ell & n \\ 0 & 0 & 0 \end{pmatrix}^2 \frac{\min\{r,s\}^n}{\max\{r,s\}^{n+1}}.
\end{equation}
\end{proposition}

\noindent\emph{Remark.} 
An easy consequence of this proposition is that for all $\ell,\ell' \in \N_0$
\begin{equation}
\label{exchange kernel}
\frac{1}{(4\pi)^2} \int_{(\mathbb{S}^2)^2} \frac{P_{\ell}({\bf \hat x} \cdot {\bf \hat y})P_{\ell'}({\bf \hat x} \cdot {\bf \hat y})}{|r{\bf \hat x}-s{\bf \hat y}|} \,d\sigma({\bf \hat x},{\bf \hat y}) = \sum_{k=|\ell-\ell'|}^{\ell+\ell'} \begin{pmatrix} \ell & \ell' & k \\ 0 & 0 & 0 \end{pmatrix}^2 \frac{\min\{r,s\}^k}{\max\{r,s\}^{k+1}}.
\end{equation}
This is seen by multiplying \eqref{eqn general} with $\overline{Y_{LM}({\bf \hat x})}$, integrating over $\mathbb{S}^2$ with respect to ${\bf\hat x}$ and summing over $M=-L,\dots,L$.

\begin{proof}
Assume first that $r\neq s$. 
For fixed ${\bf \hat x} \in \mathbb{S}^2$ the series expansion 
$$
\frac{1}{|r{\bf \hat x}-s{\bf \hat y}|} = \frac{1}{r_>} \sum_{n=0}^\infty \left( \frac{r_<}{r_>} \right)^n P_n({\bf \hat x}\cdot {\bf \hat y})
$$
converges pointwise for all ${\bf \hat y}\in \mathbb{S}^2$ and thus in $L^2(\mathbb{S}^2)$ because $\sum_{n=0}^N \left( \frac{r_<}{r_>} \right)^n P_n({\bf\hat x} \cdot {\bf\hat y})$ is bounded uniformly in $N$ and ${\bf\hat y}$. We get
\begin{eqnarray*}
\frac{P_\ell({\bf\hat x} \cdot {\bf\hat y})}{|r{\bf\hat x}-s{\bf\hat y}|} & = & \frac{1}{r_>} \sum_{n=0}^\infty \left( \frac{r_<}{r_>} \right)^n P_n({\bf\hat x} \cdot {\bf\hat y}) P_\ell({\bf\hat x} \cdot {\bf\hat y}) \\
& = & \frac{1}{r_>} \sum_{n=0}^\infty \left( \frac{r_<}{r_>} \right)^n  \sum_{k=|\ell-n|}^{\ell+n} (2k+1) \begin{pmatrix} k & \ell & n \\ 0 & 0 & 0 \end{pmatrix}^2 P_k({\bf\hat x} \cdot {\bf\hat y})
\end{eqnarray*}
where we used the addition theorem
$$
P_n(z)P_\ell(z) = \sum_{k=|\ell-n|}^{\ell+n} (2k+1) \begin{pmatrix} k & \ell & n \\ 0 & 0 & 0 \end{pmatrix}^2 P_k(z).
$$
The addition theorem \eqref{LegProp} allows us to compute
\begin{eqnarray*}
\lefteqn{\frac{1}{4\pi}\int_{\mathbb{S}^2} \frac{P_{\ell}({\bf\hat x} \cdot {\bf\hat y})Y_{LM}({\bf \hat y})}{|r{\bf\hat x}-s{\bf\hat y}|} \,d\sigma({\bf\hat y})} \\
& = &  \frac{1}{r_>} \sum_{n=0}^\infty \left( \frac{r_<}{r_>} \right)^n \sum_{k=|\ell-n|}^{\ell+n} \begin{pmatrix} k & \ell & n \\ 0 & 0 & 0 \end{pmatrix}^2 \sum_{m=-k}^k Y_{km}({\bf\hat x}) \int_{\mathbb{S}^2} \overline{Y_{km}({\bf\hat y})} Y_{LM}({\bf\hat y}) \,d\sigma({\bf\hat y}) \\
& = &  Y_{LM}({\bf\hat x}) \sum_{n=0}^{\infty} \begin{pmatrix} L & \ell & n \\ 0 & 0 & 0 \end{pmatrix}^2 \frac{\min\{r,s\}^n}{\max\{r,s\}^{n+1}}.
\end{eqnarray*}
The desired equation for $r \neq s$ now follows from the fact that the Wigner 3j-symbols vanish unless $|L-\ell| \leq n \leq L+\ell$. The case $r=s$ can be derived from the above result by choosing a sequence $r_n \downarrow s$. Clearly, $\frac{1}{|r_n {\bf\hat x}-s{\bf\hat y}|} \uparrow \frac{1}{|s{\bf\hat x}-s{\bf\hat y}|}$ for all ${\bf\hat y} \in \mathbb{S}^2 \setminus \{{\bf\hat x}\}$ and $\frac{1}{|{\bf\hat x}- {\bf\hat y}|}$ is integrable with respect to ${\bf\hat y} \in \mathbb{S}^2$. Hence Lebesgue's Dominated Convergence Theorem may be used to see that the formula is also true for $r=s$.
\end{proof}

Let us turn to the derivation of $\EE^{RHF}$. 
If $f_1,\dots,f_{s_0}$ are in $H^1_0(\R_+)$, then the functions $\ph_{jm\sigma}$ defined by \eqref{restr form} are orthonormal in $L^2(\R^3;\C^2)$, and $\ph_{jm\sigma} \in H^1(\R^3;\C^2)$ by Hardy's inequality
\begin{equation}\label{Hardy}
\int_{\R_+} \frac{|f(r)|^2}{r^2} \,dr \leq 4 \int_{\R_+} |f'(r)|^2 \,dr
\end{equation}
for $f \in H^1_0(\R_+)$.
Using the addition theorem \eqref{LegProp}, the corresponding density matrix $\tau$ and electronic density $\rho$ take the form
\begin{eqnarray}
\tau(x,y) & = & \delta_{\mu_x \mu_y} \sum_{j=1}^{s_0} \frac{2\ell_j+1}{4\pi}  \frac{f_j(|{\bf x}|)}{|{\bf x}|}\frac{\overline{f_j(|{\bf y}|)}}{|{\bf y}|} P_{\ell_j}({\bf \hat x} \cdot {\bf \hat y}), \label{tau explicit} \\
\rho(x) & = & \sum_{j=1}^{s_0} \frac{2\ell_j+1}{4\pi} \frac{|f_j(|{\bf x}|)|^2}{|{\bf x}|^2}.
\end{eqnarray}
Here we abbreviate ${\bf \hat x} := {\bf x}/ |{\bf x} |$ for all $0 \neq {\bf x} \in \R^3$. 
If the general Hartree--Fock functional \eqref{HF functional} is evaluated at the functions $\ph_{jm\sigma}$, the only term, which is not trivially computed, is the exchange term:
\begin{eqnarray*}
\int\!\!\!\int \frac{|\tau(x,y)|^2}{|x-y|} \,dx\,dy & = & 2 \sum_{j,k=1}^{s_0} \frac{(2\ell_j+1)(2\ell_k+1)}{(4\pi)^2} \int_{(\R_+)^2} \,dr\,ds \overline{f_j(r)f_k(s)}f_k(r)f_j(s) \times \\
& & \times \int_{(\mathbb{S}^2)^2} \,d\sigma({\bf \hat x}, {\bf \hat y}) \frac{P_{\ell_j}({\bf \hat x} \cdot {\bf \hat y}) P_{\ell_k}({\bf \hat x} \cdot {\bf \hat y})}{|r{\bf \hat x} - s {\bf \hat y}|}. 
\end{eqnarray*}
Using \eqref{exchange kernel}, the form of \eqref{rest energy} follows from the choice
\begin{equation}\label{Uform}
U_{\ell \ell'}(r,s) = \sum_{k=|\ell-\ell'|}^{\ell+\ell'} \begin{pmatrix} \ell & \ell' & k \\ 0 & 0 & 0  \end{pmatrix}^2 \frac{\min\{r,s\}^k}{\max\{r,s\}^{k+1}}.
\end{equation}

\section{Appendix}\label{sec:lm}

\begin{lemma}\label{lm:U-ell}
Let $\ell,\ell' \in \N_0$, and $r,s >0$. Then the functions $U_{\ell\ell'}$ defined by \eqref{Uform} obey:
\begin{itemize}
\item[(U1)] $U_{\ell\ell'}(r,s) = U_{\ell' \ell}(r,s) = U_{\ell \ell'}(s,r)$,
\item[(U2)] $0 \leq U_{\ell \ell'}(r,s) \leq \max\{r,s\}^{-1},$
\item[(U3)] $\displaystyle U_{\ell\ell}(r,s) \geq \frac{1}{2\ell +1} \frac{1}{\max\{r,s\}}$,
\item[(U4)] For all $g \in H^1_0(\R_+)$ the integral kernels $g(r) U_{\ell\ell'}(r,s) \overline{g(s)}$ define non-negative Hilbert--Schmidt operators on $L^2(\R_+)$.
\end{itemize}
\end{lemma}

\begin{proof}
(U1) and (U3) are obvious from the explicit representation of $U_{\ell\ell'}(r,s)$ and 
$$\begin{pmatrix} \ell & \ell & 0 \\ 0 & 0 & 0 \end{pmatrix}^2 = \frac{1}{2\ell+1}.$$ (U2) The positivity of $U_{\ell \ell'}$ is clear, the upper bound can be proved using \eqref{exchange kernel}, \eqref{LegProp} and Cauchy--Schwarz:
\begin{eqnarray*}
\lefteqn{U_{\ell \ell'} (r,s) = \frac{1}{(4\pi)^2} \int_{(\mathbb{S}^2)^2} \frac{P_{\ell}({\bf \hat x} \cdot {\bf \hat y})P_{\ell'}({\bf \hat x} \cdot {\bf \hat y})}{|r{\bf \hat x}-s{\bf \hat y}|} \,d\sigma({\bf \hat x},{\bf \hat y})} \\
& \leq & \frac{1}{(2\ell+1)(2\ell' +1)} \sum_{m=-\ell}^{\ell} \sum_{m'=-\ell'}^{\ell'} \left( \int_{(\mathbb{S}^2)^2} \frac{|Y_{\ell m}({\bf \hat x})|^2 |Y_{\ell' m'}({\bf \hat y})|^2}{|r{\bf \hat x}-s{\bf \hat y}|} \,d\sigma({\bf \hat x},{\bf \hat y}) \right)^{1/2} \times \\
& & \quad\qquad\qquad\qquad\qquad\qquad\qquad \qquad \times \left( \int_{(\mathbb{S}^2)^2} \frac{|Y_{\ell' m'}({\bf \hat x})|^2 |Y_{\ell m}({\bf \hat y})|^2}{|r{\bf \hat x}-s{\bf \hat y}|} \,d\sigma({\bf \hat x},{\bf \hat y}) \right)^{1/2} \\
& \leq & \frac{1}{(4\pi)^2} \int_{(\mathbb{S}^2)^2} \frac{1}{|r{\bf \hat x}-s{\bf \hat y}|} \,d\sigma({\bf \hat x},{\bf \hat y}) = \frac{1}{\max\{r,s\}},
\end{eqnarray*}
where we used $2ab \leq a^2 + b^2$ and \eqref{LegProp} in the last inequality.

(U4) The integral kernels $K(r,s) := g(r) U_{\ell\ell'}(r,s) \overline{g(s)}$ are in $L^2(\R_+^2)$ by (U2) and by Hardy's inequality \eqref{Hardy}, which shows that the corresponding integral operators are Hilbert--Schmidt. Moreover, let
$$
\ph_m({\mathbf x},\mu) := \frac{g(|{\mathbf x}|)}{|{\mathbf x}|} Y_{\ell m}({\mathbf x}) \delta_{\mu,+1}, \quad m=-\ell,\dots,\ell,
$$
and
$$
\tau(x,y)  :=  \sum_{m=-\ell}^\ell \ph_m(x)\overline{\ph_m(y)} = \delta_{\mu_x,+1} \delta_{\mu_y,+1} \frac{2\ell + 1}{4\pi}  \frac{g(|{\bf x}|)}{|{\bf x}|}\frac{\overline{g(|{\bf y}|)}}{|{\bf y}|} P_{\ell}({\bf \hat x} \cdot {\bf \hat y}).
$$
Given $f \in L^2(\R_+)$, we define
$$
\ph({\mathbf x},\mu) := \frac{f(|{\mathbf x}|)}{|{\mathbf x}|} Y_{\ell' 0}({\mathbf x}) \delta_{\mu,+1},
$$
then 
$$
\int\!\!\!\int \frac{\overline{\ph(x)}\tau(x,y)\ph(y)}{|x-y|} \,dx\,dy = (2\ell + 1)\int\!\!\!\int \overline{f(r)} K(r,s) f(s) \,dr\,ds.
$$
The last equality can be computed using \eqref{eqn general} and \eqref{Uform}.
Hence, the non-negativity of the integral operator corresponding to $K$ follows from the non-negativity of the term on the left-hand side.
\end{proof}

\begin{lemma}\label{lm:technical}
\begin{itemize}
\item[(i)] For all $f \in H^1_0(\R_+)$ and $\eps >0$: $\displaystyle \sprod{f}{\frac{1}{r}f} \leq \eps \norm{f'}^2 + \frac{1}{\eps} \norm{f}^2$.
\item[(ii)] Let $s_0 \in \N$, $\ell_1,\dots,\ell_{s_0} \in \N_0$, $Z>0$, $f_1,\dots,f_{s_0} \in H^1_0(\R_+)$, and $\eps>0$. Then:
$$
\EE^{RHF}(f_1,\dots,f_{s_0}) \geq 2 \sum_{j=1}^{s_0} (2\ell_j+1) \left[ (1-Z\eps) \| f_j' \|^2 - \frac{Z}{\eps} \norm{f_j}^2 \right].
$$
\item[(iii)] Let $s_0 \in\N$, $\ell_1,\dots, \ell_{s_0} \in \N_0$, and $f_1,\dots,f_{s_0} \in H_0^1(\R_+)$. Then for all $\ell \in \N_0$: 
$$
0 \leq K_{\ell} \leq U \leq \sum_{k=1}^{s_0} (2\ell_k+1) (\|f_k'\|^2+\|f_k\|^2).
$$
\item[(iv)] Let $\ell,\ell' \in \N_0$. Then the following maps are weakly sequentially continuous on $H_0^1(\R_+)$ resp.~$H_0^1(\R_+) \times H_0^1(\R_+)$: 
\begin{eqnarray*}
f & \mapsto & \sprod{f}{\frac{1}{r}f}, \\
(f,g) & \mapsto & \form{f \otimes g}{\max\{r,s\}^{-1}}{f \otimes g}, \\
(f,g) & \mapsto & \form{f \otimes g}{U_{\ell \ell'}}{g \otimes f}.
\end{eqnarray*}
\item[(v)] The functional $\EE^{RHF}$ is weakly sequentially lower semicontinuous on $\times_{i=1}^{N} H_0^1(\R_+)$.
\end{itemize}
\end{lemma}

\begin{proof}
(i) and (iii) follow easily from the Cauchy--Schwarz and the Hardy inequalities \eqref{Hardy}, (U1), (U2), and (U4). 
To prove (ii) fix $j,k \in \{1,\dots,s_0\}$. Using Cauchy--Schwarz, (U1) and (U2) we obtain
\begin{eqnarray*}
& & \left| \int\!\!\!\int \overline{f_j(r) f_k(s)} U_{\ell_j \ell_k}(r,s) f_k(r) f_j(s) \,dr\,ds \right| \\
& \leq & \left( \int\!\!\!\int |f_j(r)|^2 |f_k(s)|^2 U_{\ell_j \ell_k}(r,s) \,dr \,ds \right)^{\frac{1}{2}} \left( \int\!\!\!\int |f_k(r)|^2 |f_j(s)|^2 U_{\ell_j \ell_k}(r,s) \,dr \,ds \right)^{\frac{1}{2}} \\
& = & \int\!\!\!\int |f_j(r)|^2 |f_k(s)|^2 U_{\ell_j \ell_k}(r,s) \,dr \,ds \leq \int\!\!\!\int \frac{|f_j(r)|^2 |f_k(s)|^2}{\max\{r,s\}}  \,dr \,ds.
\end{eqnarray*}
Therefore,
$$
\EE^{RHF}(f_1,\dots,f_{s_0}) \geq 2 \sum_{j=1}^{s_0} (2\ell_j +1) \left( \|f_j' \|^2 - Z \sprod{f_j}{\frac{1}{r} f_j} \right).
$$
The claim now follows immediately from (i).

(iv) Let $f_n \rightharpoonup f$ weakly in $H^1_0(\R_+)$. Due to the Rellich--Kondrashov theorem, $f_n$ converges to $f$ uniformly in $\R_+$.
To prove the weak continuity of the Coulomb potential we first use 
$$
\left| \sprod{f_n}{\frac{1}{r}f_n} - \sprod{f}{\frac{1}{r}f} \right| \leq 
\left| \sprod{f_n-f}{\frac{1}{r}f_n} \right| + \left| \sprod{\frac{1}{r}f}{f_n-f} \right| = (*) + (**).
$$
For $R>0$ we obtain using Cauchy--Schwarz and Hardy's inequality \eqref{Hardy}
\begin{eqnarray*}
(*) & \leq & \int_0^R \frac{|f_n(r)-f(r)| |f_n(r)|}{r} \,dr + \frac{1}{R} \int_R^\infty |f_n(r)-f(r)| |f_n(r)| \,dr \\
& \leq & \left( \int_0^R |f_n(r)-f(r)|^2 \,dr \right)^{1/2} \left( \int_0^\infty \frac{|f_n(r)|^2}{r^2} \,dr \right)^{1/2} + \frac{1}{R}\norm{f_n-f}\norm{f_n} \\
& \leq & 2 \sqrt{R} \sup_{r \in (0,R)} \{|f_n(r)-f(r)|\} \norm{f_n'} + \frac{1}{R} \left( \norm{f_n} + \norm{f} \right) \norm{f_n}.
\end{eqnarray*}
Since $\| f_n' \|$, $\| f \|,$ and $\| f_n\|$ are uniformly bounded in $n$, we can first choose $R$ large to make the second term small, then choose $n$ large to make the first term small. $(**)$ can be estimated analogously. 
The weak continuity of the other maps can be seen with a similar decomposition argument as shown above for the Coulomb potential.

(v) Let $f_j^{(n)} \rightharpoonup f_j$ weakly in $H_0^1(\R_+)$ for $j=1,\dots,N$. Clearly,
$$
\form{f_j}{-\partial_r^2 + \frac{\ell_j(\ell_j+1)}{r^2}}{f_j} \leq \liminf_{n \to \infty} \form{f_j^{(n)}}{-\partial_r^2 + \frac{\ell_j(\ell_j+1)}{r^2}}{f_j^{(n)}},
$$
since $f_j^{(n)} \rightharpoonup f_j$ in $H^1_0(\R_+)$ implies $\partial_r f_j^{(n)} \rightharpoonup \partial_r f_j$ in $L^2(\R_+)$ for the first term, and using the lemma of Fatou for the second term. The remaining terms of $\EE^{RHF}$ are weakly sequentially continuous as shown in (iv).
\end{proof}

\end{document}